\documentclass[pra,aps,twocolumn,showpacs,superscriptaddress,nofootinbib]{revtex4}

\usepackage{graphicx}
\usepackage{amssymb}
\usepackage{amsmath}

\newtheorem{theorem}{Theorem}

\newtheorem{lemma}[theorem]{Lemma}
\newtheorem{prep}[theorem]{Proposition}

\newcommand{\qed}{\hfill $\blacksquare$}

\newtheorem{defn}{Definition}

\newcommand{\ip}[2]{\langle #1|#2 \rangle}

\begin{document}

\title{Quantum Control via Geometry: An explicit example}

\author{Mile Gu}
\address {Department of Physics, University of Queensland, St
Lucia, Queensland 4072, Australia.}

\author{Andrew Doherty}
\address {Department of Physics, University of Queensland, St
Lucia, Queensland 4072, Australia.}

\author{Michael A. Nielsen}
\address {Department of Physics, University of Queensland, St
Lucia, Queensland 4072, Australia.}
\address{Perimeter Institute for Theoretical Physics, Waterloo, ON N2L 2Y5, Canada.}

\date{\today}

\begin{abstract}We explicitly compute the optimal cost for a class of example
problems in geometric quantum control. These problems are defined by a
Cartan decomposition of $su(2^n)$ into orthogonal subspaces $\mathfrak{l}$ and $\mathfrak{p}$ such
that $[\mathfrak{l},\mathfrak{l}] \subseteq \mathfrak{p}, [\mathfrak{p},\mathfrak{l}] = \mathfrak{p}, [\mathfrak{p},\mathfrak{p}] \subseteq \mathfrak{l}$. Motion in the $\mathfrak{l}$ direction are assumed to have negligible cost, where motion in the $\mathfrak{p}$ direction do not. In the special case of two qubits, our results correspond to the minimal interaction cost of a given unitary.
\end{abstract}


\maketitle

\section{Introduction}

Characterizing the difficulty of synthesizing particular quantum interactions has generated considerable interest in recent years due to its practical applications in quantum computation. From the perspective of optimal control, it determines the optimal way to construct a desired quantum interaction with a limited set of tools \cite{Ernst90a,Glaser98a,Warren97a,Carlini05a,Carlini06a}. From the perspective of quantum circuits, it expresses the minimal number of basic gates required to build up a given algorithm \cite{Nielsen00a}.

These perspectives result in different characterizations of complexity. In optimal control, a unitary is hard if it is costly to synthesize with available interactions. In quantum circuits, a hard unitary requires a large number of the basic available gates. Recent work by Nielsen et al. shows that for certain control problems, both characterizations are polynomially equivalent\footnote{The are some technical caveats to this equivalence related to approximate
versus exact implementation.  See \cite{Nielsen06d} for details} \cite{Nielsen06d}.

This equivalence motivates the application of continuous geometrical methods to quantum circuits, and in the special case where the basic gates are single and two qubit unitaries, quantum complexity  \cite{Nielsen06b,Nielsen06c}. In this formulation, each unitary operator corresponds to a point in a particular Riemmannian manifold. The metric is engineered such that the minimal distance between a unitary operator $U$ and the identity $I$ corresponds to the minimal cost of synthesizing $U$. This approach allows us to apply mathematical techniques cultivated over many decades to a significantly newer field.

Prior work in quantum optimal control has mostly dealt with systems that evolve under a specific drift Hamiltonian (See for example, \cite{Bennet02a,Khaneja01a,Childs03d,Boscain05a}). However, all entangling operations are equivalent modulo local interactions \cite{Dodd02a,Wocjan02c,Dur00b,Bennet02a,Jones99a} and hence no particular operation should be favored in a model compliant with the spirit of quantum complexity. This motivates the treatment of interaction Hamiltonians as a physical resource, where they are all assigned equal cost.


In this paper, we consider a class of quantum control problems where the space of Hamiltonians is divided into two orthogonal subspaces, the application of Hamiltonians in one subspace incurs negligible cost compared to the other. Provided these subspaces satisfy the conditions of a Cartan decomposition (see below), geometrical methods may be used to construct a general solution.

While the general class of systems solved in this paper has not been analyzed in previous literature, it encompasses a number of previously studied systems. In the special case of a single qubit, our result provides an alternative characterization of single qubit time-optimal control \cite{Khaneja01a}. In the case of $2$-qubits, our solution coincides with the interaction cost of a two qubit unitary \cite{Vidal02c}, minimized over all possible drift Hamiltonians.






\section{Background and Definitions}
In this section, we introduce some of the necessary background and notation that will be used in the paper. We assume the reader is familiar with the basic notions of Riemannian geometry, Lie algebras and quantum circuits (E.g ~\cite{Lee97a,Georgi99a,Nielsen00a}).

Consider an $n$-qubit system. The space of traceless Hamiltonians $H \in su(2^n)$ on this system forms a vector space under the trace inner product $\ip{A}{B} = \mathrm{tr}(AB)$. This space is spanned by the product operator basis $\prod_{j=1}^n \sigma_{j,k}$, where $\sigma_{j,k}$ denotes the action of applying the Pauli interaction $\sigma_{j} \in \{I,\sigma_x,\sigma_y,\sigma_z\}$ to the $k^{th}$ qubit.

The quantum control problem is defined as follows. We wish to synthesize a given $n$-qubit unitary $U \in SU(2^n)$ by the application of some Hamiltonian $H(t) \in su(2^n)$. We define a cost function $C: SU(2^n) \times su(2^n) \rightarrow \mathbb{R}$ such that the application of $H$ for duration $dt$ on a unitary $U $ incurs cost $C(U, H) dt$. Formally, the system is governed by the Schr\"{o}dinger equation
\begin{equation}\label{eqn:evo}
\frac{dU}{dt} =  - i H(t) U(t) \qquad U(0) = I \qquad U(T) = U.
\end{equation}
We aim to find the $H(t)$ on interval $[0,T]$ such that the total cost $D(I,U) = \int_0^T C(U, H) dt$ is minimized.

In this paper, we analyze such problems using the geometrical approach \cite{Nielsen06d}. Each unitary $U \in SU(2^n)$ corresponds to a point in the Riemannian manifold $\mathcal{N} = SU(2^n)$, and each Hamiltonian describes a vector in $T\mathcal{N}$, the tangent space of $\mathcal{N}$.  Distances on $\mathcal{N}$ are defined by $C: \mathcal{N} \times T\mathcal{N} \rightarrow \mathbb{R}$. The minimal cost $D(I,U)$ coincides with the minimal distance between $I$ and $U$.

We focus on a class of quantum control problems that split the space of Hamiltonians $su(2^n)$ into two vector subspaces, $\mathfrak{l}$ and $\mathfrak{p}$. Hamiltonians in $\mathfrak{l}$ have negligible cost, while those in $\mathfrak{p}$ do not. In addition, $\mathfrak{l}$ and $\mathfrak{p}$ satisfy the set of commutation relations that define a Cartan decomposition \cite{Khaneja01b}:
\begin{equation}\label{eqn:com_relation}
[\mathfrak{l},\mathfrak{l}] \subseteq \mathfrak{l} \qquad [\mathfrak{p},\mathfrak{l}] = \mathfrak{p}, \qquad [\mathfrak{p},\mathfrak{p}] \subseteq \mathfrak{l}.
\end{equation}
We refer to such problems as Cartan control problems.

\begin{defn}[Cartan control problem]
A Cartan control problem on an $n$-qubit system is defined as follows. Let $\mathfrak{l}$ and $\mathfrak{p}$ be subspaces of $su(2^n)$ that satisfy (\ref{eqn:com_relation}). Define $P_\mathfrak{l}$ and $P_\mathfrak{p}$ as their respective projection operators. The application of a Hamiltonian $H$ for time $dt$ incurs cost $C(U,H) = \sqrt{\ip{H}{\mathcal{\tilde{G}}H}}dt$, where
\begin{equation}
\mathcal{\tilde{G}} = \epsilon P_\mathfrak{l} +  P_\mathfrak{p}, \qquad \epsilon \ll 1.
\end{equation}
Given an unitary $U \in SU(2^n)$, we wish to find $H(t)$ on $[0,T]$ that minimizes $D(I,U) = \int_0^T C(U(T), H(t)) dt$ subject to (\ref{eqn:evo}). Alternatively. this problem can be regarded as computing the distance between $I$ and $U$ on the manifold $N = SU(2^n)$ subject to the metric $C$.
\end{defn}

The $2$-qubit system, where we wish synthesize $U \in SU(2^n)$ with the minimal amount of non-local interactions, is a special case of this problem. Here, $\mathfrak{l}$ is the vector space of single-qubit Hamiltonians and $\mathfrak{p}$ is the vector space of all directions orthogonal to $\mathfrak{l}$. The resulting Cartan control problem neglects the cost of all single-qubit interactions, and thus $D(I,U)$ is a measure of the minimal amount of interactions required to synthesize $U$. In fact, it coincides with the interaction cost of $U$ \cite{Vidal02c}, when minimized over all possible drift Hamiltonians.

The physical interpretation of $n$-qubit Cartan control problems for $n > 2$ is not as transparent. Although there exists a decomposition such that all single-qubit interactions are contained in $\mathfrak{l}$, $\mathfrak{l}$ will invariably also contain interactions involving an unbounded number of qubits. Therefore, the condition $\epsilon \rightarrow 0$ implies that in addition to local interactions, certain non-local interactions can also be applied at negligible cost. Although the solution for these cases does not have a direct physical application, other than provide a lower bound on complexity, it shows how geometrical methods are well adapted to solving a general class of problems that scale with $n$.

\section{Solution to the Cartan Control Problem}
In this section, we solve the Cartan control problem using geometrical methods \cite{Nielsen06d}. Before approaching the problem directly, we illustrate the intuition behind our approach by a simple example.

Consider a cylindrical surface of unit radius $N = \mathbb{R}\times[0,2\pi)$ parameterized by standard cylindrical coordinates, $z$ and $\theta$ and the naturally induced metric $C_N(z, \theta, dz,d\theta) = \sqrt{dz^2 + d\theta^2}$. Suppose we wish to find the minimal distance between two points, $\mathbf{x} = (0,0)$ and $\mathbf{y} = (0,\pi/2)$, it is clear that geodesics between the two points are non-unique since the surface wraps around itself. We circumvent this difficulty by introducing a second manifold $\mathcal{M} = \mathbb{R}^2$ with the standard Euclidean metric $C_M(p,q,dp,dq) = \sqrt{dp^2 + dq^2}$, together with a mapping $U: \mathcal{M} \rightarrow \mathcal{N}$ of the form $U(p,q) = (p, q\mod 2\pi)$. If we define $[\mathbf{x}]$ and $[\mathbf{y}]$ as the pre-image of $\mathbf{x}$ and $\mathbf{y}$ with respect to $U$. i.e: $[\mathbf{x}] = \{0,2j\pi\}$, $[\mathbf{y}] = \{0,(2k+\frac{1}{2})\pi)\}$ $j,k \in \mathbb{Z}$, then the distance $d_N(\mathbf{x},\mathbf{y})$ on $\mathcal{N}$ coincides with the minimal distance between the sets $[\mathbf{x}]$ and $[\mathbf{y}]$ on $\mathcal{M}$, i.e: $\pi/2$. The following lemma states this more generally:

\begin{lemma}\label{lem:distance_equiv}Let $\mathcal{M}$ and $\mathcal{N}$ be Riemannian manifolds with distance measures $C_M$ and $C_N$. Denote the distance between two points on $\mathcal{M}$ and $\mathcal{N}$ by $d_M(\cdot,\cdot)$ and $d_N(\cdot,\cdot)$ respectively.

Let $U: \mathcal{M}\rightarrow \mathcal{N}$ be a smooth map that preserves the distance, i.e: $C_M(q,\mathbf{v}) = C_N(U(q),U^*(\mathbf{v}))$, where $U^*$ is the pushforward of $U$. Define $[x] = \{\mathbf{p}: \phi(\mathbf{p}) = \mathbf{x}\}$, $[y] = \{\mathbf{q}: \phi(\mathbf{q}) = \mathbf{y}\}$ as the pre-image of $x,y \in \mathcal{N}$ respectively. If $\mathbf{p} \in [x]$, and $\mathbf{q} \in [y]$, then
\begin{align}\nonumber
d_M\left([x],[y]\right) & \equiv \min_{\mathbf{q} \in [y]} d_M(\mathbf{p},\mathbf{q})\\& = \min_{\mathbf{p} \in [x]} d_M(\mathbf{p},\mathbf{q})= d_N(\mathbf{x},\mathbf{y}),
\end{align}
where $d_M(\mathbf{p},\mathbf{q})$ denotes the distance between $\mathbf{p}$ and $\mathbf{q}$.
\end{lemma}

\begin{figure}[htp]
\centering
\includegraphics[width=0.4\textwidth]{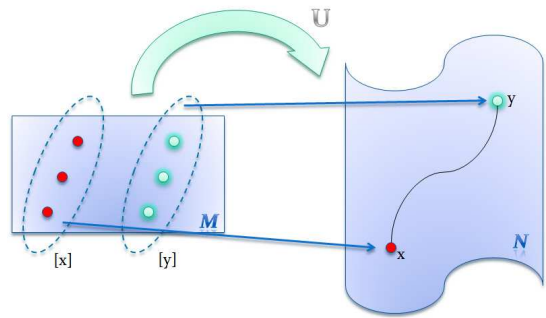}
\caption{(Color online)If $[\mathbf{x}]$ and $[\mathbf{y}]$ in $M$ are the pre-images of $\mathbf{x}$ and $\mathbf{y}$ in $N$, then $d_M([x],[y]) = d_N(x,y)$ provided the mapping $U: M \rightarrow N$ is smooth.}\label{fig:erptsqfit}
\end{figure}

 \begin{proof} Suppose $d_M([x],[y]) = k$, then there exists a curve $\gamma \subset \mathcal{M}$ that connects some $\mathbf{p} \in [x]$ and $\mathbf{q} \in [y]$ of length $k$. Clearly its image  $\Gamma(t) = U(\gamma(t))$ is a curve from $\mathbf{x}$ to $\mathbf{y}$ of length $k$ in $\mathcal{N}$. Thus $d_N(\mathbf{x},\mathbf{y}) \leq d_M\left([x],[y]\right)$.

Now suppose $d_N(\mathbf{x},\mathbf{y}) = k$, such that there exists a curve $\Gamma(t) \in \mathcal{N}$, $\Gamma(0) = \mathbf{x}$, $\Gamma(1) = \mathbf{y}$ of length $k$. Given any $\mathbf{p} \in [x]$, we show that there exists a $\mathbf{q} \in [y]$ such that $d_N(\mathbf{p},\mathbf{q}) \leq k$ by constructing a curve $\gamma$ from $\mathbf{p}$ to $\mathbf{q}$. Let $l$ be a large integer and $t_j = jk/l$, $j = 0, \ldots, m$, set
\begin{align}
\gamma(t) = \mathbf{v}_{t_j}t + \gamma\left(t_j\right) \qquad  t_j  \leq t < t_{j+1}
\end{align}
where $m$ is a large integer, and $\mathbf{v}_{t_j}$ satisfies
\begin{equation}
U^* \mathbf{v}_{t_j} = \frac{d\Gamma}{dt}(t_j)
\end{equation}
Clearly in the limit $n \rightarrow \infty$:
\begin{equation}
\mathrm{Length}(\gamma(t)) = \lim_{l \rightarrow \infty} \sum_j C_M \left(\gamma\left(t_j\right), \mathbf{v}_{t_j} \right) = k
\end{equation}
Hence $d_N(\mathbf{x},\mathbf{y}) \geq \min_{\mathbf{q} \in [x]} d_M(\mathbf{p},\mathbf{q})$. Combining the two results gives the desired equivalence. Symmetry implies $d_N(\mathbf{x},\mathbf{y}) \geq \min_{\mathbf{p} \in [y]} d_M(\mathbf{p},\mathbf{q})$, which establishes the desired result. \end{proof} \qed

To compute distances on $SU(2^n)$, we define an Euclidean manifold $\mathcal{M} = \mathbb{R}^{4^n -1}$. Denote its coordinates by $\mathbf{q} = (q_1,\ldots,q_{4^n-1})$ and tangent vectors by $\mathbf{v}$. We wish to find a suitable metric $C_M(\mathbf{q},\mathbf{\mathbf{v}}) = \sqrt{\ip{\mathbf{\mathbf{v}}}{G \mathbf{\mathbf{v}}}}$, together with a distance preserving map $U$ such that Lemma \ref{lem:distance_equiv} is applicable. There are many possible choices, of which we ideally select one where $G$ has a simple form.


%

The Cartan decomposition is one such candidate \cite{Khaneja01b}. Let $\mathfrak{z}$ be a maximally commuting subspace of $\mathfrak{p}$, then any unitary can be decomposed into
\begin{equation}\label{eqn:decompfirst}
U(H_1,H_2,H_3) = e^{-iH_1} e^{-i H_2} e^{-iH_3},
\end{equation}
where $H_1,H_3 \in \mathfrak{l}$ and $H_2\in \mathfrak{z}$. The vector of matrices $(H_1, H_2, H_3)$ completely specify $U$. We view this as a cartesian plane, $(\mathbf{q_1},\mathbf{q_2}, \mathbf{q_3}) \in \mathcal{M}$, where $\mathbf{q_i}$ is the vectorization of $H_i$ with respect to some orthonormal basis $B_{j,i}$, i.e. $H_i = \sum_j q_{j,i} B_{j,i}$ and $q_{j,i} = \mathrm{Tr}[H_iB_{j,i}]$. (\ref{eqn:decompfirst}) then defines the desired coordinate map.

The second step is to compute the metric $C_M$ on $\mathcal{M}$ such that $C_M(q,\mathbf{v}) = C_N(U(q),U^*(\mathbf{v}))$.  The matrix $G$ can be represented by a $3 \times 3$ matrix of superoperators, $G_{i,j}$, such that a perturbation $(\Delta H_1,\Delta H_2, \Delta H_3)$ has length $\Delta \sum_{i,j} \sqrt{\ip{H_i}{G_{ij} H_j}}$. The properties of $G$ can be characterized:

\begin{lemma}\label{lemma:local_metric}Let $\mathcal{M} = \mathbb{R}^{4^ n-1}$ be a Riemmannian manifold with metric $C_M(\Delta H_1, \Delta H_2, \Delta H_3) = \Delta \sum \sqrt{\ip{H_i}{G_{i,j}H_j}}$, and $U: \mathcal{M} \rightarrow SU(2^n)$ be as defined by (\ref{eqn:decompfirst}). If $C_M(q,\mathbf{v}) = C_N(U(q),U^*(\mathbf{v}))$, then $G$ has the form
\begin{equation}
G =  \left( \begin{array}{ccc}
\epsilon \mathrm{BCH}^\dag_L \mathrm{BCH}_L &  &  \\
 &  \textbf{I} &  \\
 & &  \epsilon \mathbf{A}(H_2) + \mathbf{B}(H_2))
\end{array}\right).
\end{equation}
$\mathbf{A}(H_2)$ and $\mathbf{B}(H_2)$ are $H_2$ dependent operators that satisfy $\mathbf{A}(0) = I$, $\mathbf{B}(0) = 0$, $(I + \mathbf{A}(H_2)) > 0$ and $\mathbf{B}(H_2) > 0$; and $\mathrm{BCH}$ denotes the Baker-Campbell-Hausdorff operator, which satisfies $\exp \{-i(C + D)\} = \exp\{-i \mathrm{BCH}_C(D)\}\exp\{-i D\} + O(|D|^2)$.
\end{lemma}

\begin{proof} From the BCH equation, $e^{\Delta A} e^{\Delta B} = e^{\Delta A +\Delta B} + O(\Delta^2)$, thus $\mathcal{G}_{ij} = 0$ for $i \neq j$. The remaining components can be computed by considering individual perturbations. The details are purely technical, and are left to the appendix.\qed
\end{proof}

To simplify the notation, we set $L = H_1$, $Z = H_2$ and $M = H_3$, so that a point on $\mathcal{M}$ is denoted by the 3-tuple $(L,Z,M)$. While the explicit forms of $\mathbf{A}$ and $\mathbf{B}$ are complex, the metric greatly simplifies when $\epsilon \rightarrow 0$.

\begin{lemma}\label{prep:triangle}In the limit $\epsilon \rightarrow 0$, the cost of synthesizing a unitary $U$ is given by
\begin{equation}
D(I,U) = \min_{Z: U(L,Z,M) = U} |Z|
\end{equation}
where $|Z| = \sqrt{\mathrm{Tr}(Z^2)}$ is the trace norm of $Z$ and $U$ is as defined in (\ref{eqn:decompfirst}).
\end{lemma}

\begin{figure}[htp]
\centering
\includegraphics[width=0.4\textwidth]{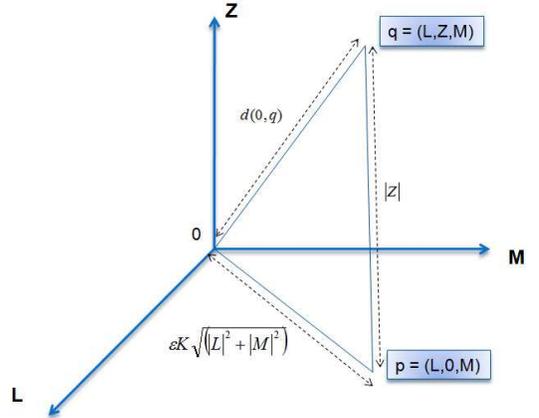}
\caption{(Color online)$d(\mathbf{p},\mathbf{q}) = |Z|$ and $d(\mathbf{0},\mathbf{q}) \rightarrow 0$ as $\epsilon \rightarrow 0$. Thus $d(\mathbf{0},\mathbf{p}) \rightarrow |Z|$ as $\epsilon$ approaches $0$. \label{fig:triangle}}
\end{figure}

\begin{proof} Consider the triangle with vertices $\mathbf{0} = (0,0,0)$, $\mathbf{p} = (L,0,M)$ and $\mathbf{q} = (L,Z,M)$ (Fig \ref{fig:triangle}). Let $C$ be a constant such the operator norm $|\mathrm{BCL}_L| < C$ and $K = \max (C,1)$. The length of straight line from the origin to $\mathbf{p}$ is bounded above by $\epsilon K \sqrt{(|L|^2 + |M|^2)}$, thus so is $d(\mathbf{0},\mathbf{p})$, the distance from the origin to $\mathbf{p}$. Two triangle inequalities then bound $d(\mathbf{0},\mathbf{q})$ from above and below:
\begin{align}\nonumber
d(\mathbf{0},\mathbf{q}) \leq d(\mathbf{p},\mathbf{q}) + d(\mathbf{0},\mathbf{p})  & \leq |Z| + \epsilon K \sqrt{(|L|^2 + |M|^2)}\\ \nonumber
d(\mathbf{0},\mathbf{q}) \geq d(\mathbf{p},\mathbf{q}) - d(\mathbf{0},\mathbf{p}) &\geq |Z| - \epsilon K \sqrt{(|L|^2 + |M|^2)}
\end{align}
In the limit $\epsilon \rightarrow 0$, $d(\mathbf{0},\mathbf{q}) = |Z|$. Application of Lemma \ref{lem:distance_equiv} gives the required result.
\end{proof} \qed

We now have an explicit characterization of distances on the coordinate manifold $M$. The final step is to determine the pre-image of a given unitary in $SU(2^n)$. We use a variation of the technique developed in \cite{Childs03d}. Let the Cartan decomposition of $U$ be as in (\ref{eqn:decompfirst}), the properties of Cartan decompositions allow us to choose a basis such that the matrix representation of $U$ can be expressed as
\begin{equation}\label{eqn:matrix_decomp}
U = A D B^T,
\end{equation}
where $A = e^{iL}$, $B^T = e^{iM}$ are orthogonal, and $D = e^{iZ}$ is diagonal \cite{Helgason01a}. $D^2$ is the diagonalisation of $U^TU$ and is hence unique up to permutation of its diagonal elements. We use this decomposition to find an explicit expression for $D(I,U)$:

\begin{theorem}\label{thrm:endresult}Consider the $n$-qubit Cartan control problem. The minimal cost required to synthesize a unitary $U$ with Cartan decomposition $e^{iL}e^{iZ}e^{iM}$ is
\begin{equation}
D(I,U) = \sqrt{\min_{y \in \mathcal{L}} |\mathbf{eig}(Z) - \mathbf{y}|^2}
\end{equation}
where $\mathcal{L}$ is a lattice defined by the set of points:
\begin{equation}
\mathcal{L} = \{(m_1, m_2, \ldots m_{2^n})\pi: \,\, \sum m_k = 0, \,\, m_k \in \mathbb{Z}\}
\end{equation}
\end{theorem}

\begin{proof}Since $D$ is diagonal, we can describe it by a vector $\mathbf{x} = (x_1,\ldots,x_{2^n})$ such that the diagonal elements of $D$ take the form $e^{i x_k}$. In addition, we know $D = e^{iZ}$ for some $Z = \sum_j z_j B_j \in \mathfrak{z}$, where $B_j$, $j = 1,\ldots 2^n-1$ is an orthonormal basis for $\mathfrak{z}$.

Let $\mathcal{A}$ be the mapping that takes the vector $\mathbf{x}$ to $\mathbf{z}$, the vector representation of $Z$ in the $B_j$ basis, i.e:
\begin{equation}
\mathbf{eig}\left\{[\mathcal{A}\mathbf{x}]^k B_k\right\} = \{x_1,x_2,\ldots, x_{2^n}\}
\end{equation}
To see that $\mathcal{A}$ is an isometry, i.e., $|\mathcal{A}\mathbf{x}|^2 = |\mathbf{x}|^2$, we note that in our particular representation (\ref{eqn:matrix_decomp}), $Z$ is diagonal, and thus $\mathbf{eig}\left\{[\mathcal{A}\mathbf{x}]^k B_k\right\} =  \mathbf{diag}\left\{[\mathcal{A}\mathbf{x}]^k B_k\right\}$.

Define $\overrightarrow{B}_k = \mathbf{diag}(B_k)$ as the vector formed from the diagonal elements of $B_k$, then $[A\mathbf{x}]^k \overrightarrow{B}_k = \mathbf{x}$. Let $\mathcal{B} = [\overrightarrow{B}_1,\overrightarrow{B}_2,\ldots,\overrightarrow{B}_k]$ be the matrix whose columns are the elements of $B_k$, then the equation can be rewritten as $\mathcal{B} \mathcal{A} \mathbf{x}   = \mathbf{x}$. Since $\mathcal{B}$ is orthonormal, $\mathcal{A} = \mathcal{B}^{-1}$ must be also, and hence preserve the Euclidean norm. So
\begin{align}\nonumber
D(I,U) & = \min \{|Z|: \, e^{L}e^{Z}e^{M} = U\},\\
& = \min \{|\mathbf{x}|: \, e^{L}e^{[\mathcal{A}\mathbf{x}]^k \sigma_k}e^{M} = U\}.
\end{align}
Since permutations preserve the Euclidean norm, the only freedom in $x_k$ that we need to minimize over is addition by multiples of $\pi$. Thus given one particular decomposition of a given unitary $U$, $\mathbf{eig}(H_2)$ gives one possible $\mathbf{x}$. The set of all vectors (permutations excluded) that generate $U$ is given by $\{\textbf{x} + \mathbf{l}: \,  \mathbf{l} \in \mathcal{L} \}$. The result follows. \qed
\end{proof}

Given a unitary $U \in SU(2^n)$, the above theorem offers a systematic method to solve for the minimal cost required to synthesize $D(I,U)$.

\section{The Single Qubit Control Problem}
In this section, we illustrate our result by applying it to the special case of single qubit optimal control. We wish to synthesize a particular spin $\frac{1}{2}$ interaction $U \in SU(2)$. Application of magnetic fields in one particular direction (say $x$) incur negligible cost, while all orthogonal directions require unit cost, i.e.
$\mathfrak{l} = \mathrm{Span}(\sigma_x)$ and $\mathfrak{p} = \mathrm{Span}(\sigma_{z,y})$.

This problem is a slight variation of the single qubit time-optimal control problem solved in \cite{Khaneja01a}. More precisely, it corresponds to the case of a system that evolves under a constant magnetic field described by the Hamiltonian $H_d = \sqrt{2} \sigma_z$. We wish to synthesize a unitary $U$ in minimal time, given the ability synthesize magnetic fields in the $x$ direction of arbitrary strength, or reverse the direction of $H_d$.

\begin{prep}Let $U \in SU(2)$. Suppose we are given one particular decomposition
\begin{equation}\label{eqn:decomp}
U = \exp(-i x \sigma_x) \exp(-i z \sigma_z) \exp(-i y \sigma_x),
\end{equation}
where $\sigma_x$, $\sigma_z$ denote standard Pauli matrices, then $D(I,U) = \frac{1}{\sqrt{2}}\min_{m\in \mathbb{Z}}\{z - 2 m\pi\}$. In particular, $D(I,U) = \frac{|z|}{\sqrt{2}}$ for $z \in [-\pi,\pi]$.
\end{prep}

\begin{proof}Note that $z \sigma_z$ has eigenvalues $\pm \frac{z}{2}$, and $\mathcal{L} = \{ (m\pi,-m\pi):\,  m \in \mathbb{Z}\}$. So $D(I,U) = \sqrt{2 \min_{m \in \mathbb{Z}} | z/2 - m \pi|^2 }$. The result follows directly. \qed \end{proof}

In \cite{Khaneja01a}, a slightly different result where $D(I,U) = \frac{|z|}{\sqrt{2}}$ for $z \in [0,2 \pi]$ is obtained. The deviation results from our extra assumption that the direction of $H_d$ can be reversed. The KGB result requires the unique decomposition such that $z \in [0,2 \pi]$, whereas our result applies to any decomposition that satisfies (\ref{eqn:decomp}).

\section{Conclusion}

The geometrical approach provides a useful alternative to more algebraic methods \cite{Bennet02a,Khaneja01a,Childs03d}. In this paper, we have demonstrated how we can use it to characterize the general Cartan control problem. In the single-qubit case, our result solves a slight variation of single-qubit time optimal control, and provides a second perspective to \cite{Khaneja01a}. In the two qubit case, it characterizes the minimal amount of non-local interactions required to synthesize a given interaction.

The general $n$-qubit Cartan control problem that we have described does not have direct physical application, because the class of Hamiltonians assumed to be `easy' to apply is too broad to be realistic. However, our results do show an instance where the geometric formalism can be applied to systems of arbitrary size. By reducing the complex space of unitary operations into a cartesian coordinate system with a suitably appropriate metric, we circumvent much of the technical difficulties in algebraically intensive methods.

The geometrical method outlined can convert any quantum control problem into a minimization of distances between two sets in cartesian space with a suitably defined metric. This allows analytical solutions in special cases, such as the Cartan control problem. Alternatively there exists numerous numerical techniques that have been designed to solve for minimal distances on manifolds. Thus, the geometric formalism is a promising method, both for solving other problems in control theory, and for its applications in quantum complexity.

%

{\bf Acknowledgments---} M.G., A.D. and M.A.N. acknowledge the support of the Australian
Research Council (ARC). M.G. thanks Christian Weedbrook, Mark de Burgh for discussions.


\begin{appendix}

\section{Derivation of the Coordinate Metric $C_M$}
In this section, we provide a detailed outline of how the metric $C_M$ in Lemma $\ref{lemma:local_metric}$ is derived. We consider a point $(L,Z,M)$ on $\mathcal{M}$ and consider perturbations on each of $L$, $Z$ and $M$, which we denote by $P_L$, $P_Z$ and $P_M$ respectively. Firstly
\begin{align}\nonumber
U(L + \Delta P_L, Z, M) & = e^{-i(L + \Delta P_L)} e^{-iZ} e^{-iM} \\  &= \exp [-i\Delta \mathrm{BCH}_L(P_L)] U
\end{align}
Since $[L,P_L] \in \mathfrak{l}$, and $\mathrm{BCH}_L \in \mathfrak{l}$,
\begin{align}
\ip{P_L}{G P_L} &= \epsilon\ip{P_L}{\mathrm{BCH}^\dag_L \mathrm{BCH}_L P_L}\qquad \\  \Rightarrow  \mathcal{G}_{11} & = \epsilon \mathrm{BCH}^\dag_L \mathrm{BCH}_L
\end{align}
Similarly, other components of $G$ can be computed by perturbing $Z$ and $M$:
\begin{align}
U(L , Z + \Delta P_Z, M)  &= e^{-iL} e^{-i(Z + \Delta P_Z)} e^{-iL} \\& = \exp \left[ -i \Delta e^{-iL} P_Z e^{iL} \right]U
\end{align}
Noting that $e^{-iL} P_Z e^{iL} \in \mathfrak{p}$ since $[\mathfrak{l}, \mathfrak{p}] \in \mathfrak{p}$, we have
\begin{equation}
\ip{P_{Z}}{G P_{Z}} = |P_Z|^2 \qquad  \Rightarrow \qquad G_{22} = I
\end{equation}
The final perturbation is slightly more complex:
\begin{align}\nonumber
U(L , Z , M + \Delta P_M) &= e^{-iL} e^{-iZ} e^{-i(M + \Delta P_M)}\\ \nonumber
& = e^{-iL} e^{ -i \Delta e^{-iZ} \mathrm{BCH}_M(P_M) e^{iZ} } e^{-iX_2} \\ \nonumber
& = e^{ -i \Delta e^{-iL} e^{-iZ} \mathrm{BCH}_M(P_M) e^{iZ} e^{-iL} } U\\
& = e^{ -i \Delta \mathcal{J}_{L,Z,M}(P_M)} U
\end{align}
Noting that $\mathrm{BCH}_M(P_M) \in \mathfrak{p}$, and $[\mathfrak{p},\mathfrak{p}] \subset \mathfrak{l}$. We can write:
\begin{equation}
 e^{-iZ} \mathrm{BCH}_M(P_M) e^{iZ} = a(Z)Q_P(P_M) + b(Z)Q_L(P_M) \qquad
\end{equation}
Where $Q_P \in \mathfrak{p}$ and $Q_L \in \mathfrak{l}$, $a^2 + b^2 = 1$ and $b(0) = 0$. The commutation relations $[\mathfrak{l},\mathfrak{l}] \subset \mathfrak{l}$, $[\mathfrak{l},\mathfrak{p}] \subset \mathfrak{p}$ then implies that $\mathcal{J}_{L,Z,M}$ takes the same form, i.e:
\begin{equation}
\mathcal{J}_{L,Z,M}(P_M) = a(Z)Q_P(P_M) + b(Z)Q_L(P_M) \qquad
\end{equation}
Thus
\begin{equation}
\ip{P_{M}}{G P_{M}} = \epsilon a^2(Z) + b^2(Z) \qquad G_{22} = I
\end{equation}
In particular
\begin{equation}
G_{33}\mid_{Z=0} = \epsilon \textbf{A}(Z) + \textbf{B}(Z)
\end{equation}
For some positive definite operators $\textbf{A}$ and $\textbf{B}$ such that $\textbf{A}(0) = \textbf{I}$ and $\textbf{B}(0) = \textbf{0}$.
 Hence, in matrix representation, the global metric takes on a block diagonal form:
\begin{equation}
\mathcal{G} =  \left( \begin{array}{ccc}
\epsilon \mathrm{BCH}^\dag_L \mathrm{BCH}_L &  &  \\
 & \textbf{I} &  \\
 & &  \epsilon A(Z) + B(Z)
\end{array}\right)
\end{equation}
\end{appendix}

\bibliographystyle{apsrev}   


\end{document}